\documentclass[conference]{IEEEtran}
\IEEEoverridecommandlockouts

\usepackage[utf8]{inputenc} 
\usepackage[T1]{fontenc}
\usepackage{url}
\usepackage{ifthen}
\usepackage{cite}
\usepackage[cmex10]{amsmath} 
\usepackage{amssymb}
\usepackage{graphicx}
\usepackage{epstopdf}
\usepackage{epsfig}
\usepackage{multirow}
\usepackage{amsthm}
\usepackage{comment}
\usepackage{caption}
\usepackage{enumitem}
\usepackage{amsmath}
\usepackage{color}
\usepackage{bbm}
\usepackage{algorithm}
\usepackage{algpseudocode}
\usepackage{csquotes}
\usepackage{subfiles}
\usepackage{caption, multirow, makecell}
\usepackage{array}
\usepackage{caption}
\usepackage{afterpage}
\usepackage{dblfloatfix}
\usepackage{booktabs}
\usepackage{tikz}
\usepackage{longtable}
\usepackage{blkarray}
\usepackage{adjustbox}
\newtheorem{thm}{Theorem}
\newtheorem{lem}{Lemma}

\newtheorem{cor}{Corollary}
\newtheorem{example}{Example}

\newtheorem{defn}{Definition}

\newtheorem{cons}{Construction}

\newtheorem{rem}{Remark}

\def\BibTeX{{\rm B\kern-.05em{\sc i\kern-.025em b}\kern-.08em
		T\kern-.1667em\lower.7ex\hbox{E}\kern-.125emX}}
\begin{document}

\title{Function-Correcting Codes for Sum-Rank Metric}

\author{\IEEEauthorblockN{Santhi Kumari Kammila and B. Sundar Rajan, {\it Life Fellow IEEE}} 
	\IEEEauthorblockA{\textit{Department of Electrical Communication Engineering,}
		\textit{Indian Institute of Science,}
		Bangalore, India \\
		\{santhik,bsrajan\}@iisc.ac.in}
}

\maketitle

\begin{abstract}
 Function-Correcting Codes (FCCs) are a class of codes designed to protect the evaluation of a specific function of a message against channel errors at a higher level than the level of protection for the message, while requiring significantly less redundancy than conventional error-correcting codes. In this paper, we study function-correcting codes under the sum-rank metric, which is a natural generalization of both the Hamming metric and the  rank-metric and also we derive general upper and lower bounds on the optimal redundancy of FCCs in the sum-rank metric. In particular, we establish a Plotkin-like  bound for irregular-distance codes in sum-rank metric and simplify it for linear functions. Furthermore, we present explicit construction of function-correcting sum-rank metric codes (FCSRCs) for locally binary functions and sum-rank weight functions with optimal redundancy. 
 \end{abstract}


\section{Introduction}
The coding framework for function-correcting codes (FCC) is introduced in \cite{b1} to protect a specific attribute (function) of the message that is of interest to the receiver  along with the entire message with reduced redundancy. When the function to be protected is a bijective mapping, the FCC coincides with  classical error-correcting code (ECC). In \cite{b1} FCCs have been studied for the Hamming metric. We study FCCs for channels matched to the sum-rank metric calling them Function-Correcting Sum-Rank Codes (FCSRCs). These are particularly beneficial in multi-shot network coding and distributed storage systems where classical sum-rank metric codes are utilized, as FCSRCs enable reliable recovery of specific function of interest with reduced redundancy compared to classical sum-rank metric codes. Rank-metric codes and sum-rank metric codes are reviewed in the following subsection and a brief review of the state-of-art works on FCCs in various settings is given in the subsequent subsection. 
\subsection{Sum-rank metric codes}
  In \cite{b2} and \cite{b3} the authors introduce rank-metric codes as a natural and powerful framework for error control in random linear network coding. Network transmission errors are modeled as additive matrix perturbations and it is shown  that the rank of the error matrix captures the effect of adversarial errors in network-coded systems. By lifting rank-metric codes to subspace codes, a direct connection between rank distance and subspace distnace is established providing a rigorous justification for using rank-metric codes to achieve optimal error correction performance in random network coding. Gabidulin introduced a fundamental class of linear codes called Maximum Rank Distance (MRD) codes defined over extension fields that achieve the maximum possible distance under the rank-metric, now known as Gabidulin codes \cite{b4}. These are analogous to maximum distance separable (MDS) codes in the Hamming metric and  attain the Singleton bound for the rank-metric and are therefore optimal. In \cite{b4} an algebraic construction of such codes based on linearized polynomials is presented and efficient decoding algorithms for correcting rank errors are developed. This work is further extended to symmetric rank-metric codes in \cite{b5} in which the authors showed that for extension fields of characteristic 2, the field can be represented by symmetric matrices, which leads to the construction of MRD codes consisting entirely of symmetric matrices. These codes achieve maximum rank distance and their vector representations correspond to linear MRD Codes.
  
  Codes over sum-rank metric are collection of vectors of matrices and the metric is obtained by adding up the ranks of all matrices. It measures errors by summing the rank errors of multiple blocks. Sum-rank metric codes generalize Hamming and rank-metric codes. They have several applications in information theory, including multishot network coding and distributed storage systems. In \cite{b6}, the authors presented a comprehensive study of codes in the sum-rank metric, covering their theoretical foundations, constructions, bounds, and applications. In \cite{b7}, the authors studied a special class of codes called Maximum Sum-Rank Distance Codes (MSRD). The fundamental properties of sum-rank metric codes are studied in \cite{b8}. Fundamental bounds such as Singleton-type bound are derived, and optimal codes achieving these bounds such as MSRD codes are characterized and explicit constructions of them are also given.
 
 \subsection{Function-Correcting Codes (FCC): State-of-art}
  In \cite{b1}, FCCs are developed for symbol substitution channel that are matched to the Hamming  metric. FCCs over symbol pair read channels are studied  in \cite{b9} which enable reliable recovery of function of interest despite pair wise read errors. Code constructions and bounds for symbol-pair metric are also presented. This is further extended to FCC over $b$-symbol read channel in \cite{b10} where $b$ consecutive (overlapping) symbols together appear in each read. Furthermore, Plotkin-like bounds on FCCs for $b$-symbol read channel and symbol-pair read channel are provided in \cite{b11}. In \cite{b12}, tighter lower and upper bounds on redundancy of FCCs especially for a class of linear fucntions are derived. In \cite{b13}, the authors presented results on optimal redundancy and efficient code constructions for functions like Hamming weight and Hamming weight distribution functions. FCCs for locally bounded functions are studied in \cite{b14} and it is further extended to $b$-symbol read channels in \cite{b15}. Lower and upper bounds on redundancy of FCCs over finite fields are derived in \cite{b16} and it is shown that these bounds holds for all finite fields. FCCs are also studied  over a metric called homogenous distance in \cite{b17} which is defined over finite rings that generalize the Hamming and Lee metric. Similarly, a theoretical framework for FCCs for channels matched to the Lee metric has been studied in \cite{b18} and  lower and upper bounds on optimal redundancy for Function-Correcting Lee Codes(FCLCs) are proposed in \cite{b19}. In \cite{b20}, the authors introduced FCCs for insertion–deletion channels where symbols are inserted or deleted from a transmitted message and derived tight redundancy bounds along with constructions showing how to reliably compute specific functions under insertions and deletions. Recently, a more general framework has been  introduced where FCCs provide data protection also along with function protection while minimizing the redundancy compared to classical error-correcting codes\cite{b21}.
\subsection{Contributions}
The main contributions of this paper are summarized below. 
\begin{itemize}
    \item FCCs for sum-rank metric are introduced. A Plotkin-like lower bound for the optimal redundancy of FCCs for the sum-rank metric is derived. The known Plotkin-like bound for FCCs for the Hamming metric and the known Plotkin-like bound for rank-metric codes are recovered as special cases. 
    \item The derived Plotkin-like bound is further simplified for linear functions for function-correcting sum-rank metric codes (FCSRCs).
	\item Explicit construction of FCSRCs are given for sum-rank locally binary functions and sum-rank weight functions and these FCSRCs are shown to be optimal.
\end{itemize} 
\subsection{Organization}
The rest of the paper is organized as follows. In Section~\ref{preliminaries}, we review sum-rank metric codes, FCCs, and irregular-sum-rank-distance codes. In Section~\ref{General  Results on the Optimal Redundancy}, the connection between FCCs for sum-rank metric and irregular-sum-rank-distance codes, and bounds on the optimal redundancy of FCCs for sum-rank metric are given. A Plotkin-like lower bound on block length of irregular-sum-rank-distance codes is also proposed in this section. Furthermore, a Plotkin bound for linear functions is derived. In Section~\ref{Sum-rank locally binary Functions}, we apply these general results to a specific class of functions called  sum-rank locally binary functions providing explicit code construction and corresponding optimal redundancy.  In Section \ref{Sum-rank weight Functions} FCCs for a new class of functions called sum-rank weight function is studied. Finally, Section~\ref{Conclusion} concludes the paper. 
\subsection{Notations}
$\mathbb{F}_q^{m \times m}$ denotes set of all $m \times m$ matrices over $\mathbb{F}_q$ the finite field with $q$ elements. The set of all $k$ length vectors of $m \times m$ matrices is denoted by $(\mathbb{F}_q^{m \times m})^k$. The rank of a matrix $\mathbf{C} \in \mathbb{F}_q^{m \times m}$ is denoted by rk($\mathbf{C}$) and srk($\mathbf{C}$) denotes the sum-rank of the vector $\mathbf{C} \in (\mathbb{F}_q^{m \times m})^k$. Given two matrices $\mathbf{A},\mathbf{B} \in \mathbb{F}_q^{m \times m}$, $d_{rk}(\mathbf{A},\mathbf{B})$ denotes the rank distance between $\mathbf{A}$ and $\mathbf{B}$. Given two vectors of matrices $\mathbf{a}, \mathbf{b} \in (\mathbb{F}_q^{m \times m})^k$, $d_{srk}(\mathbf{a}, \mathbf{b})$ denotes the sum-rank distance between $\mathbf{a}$ and $\mathbf{b}$. For any integer $M$, we define $[M]^+ \triangleq \max\{M,0\}$ and let $[M] \triangleq \{1,\ldots,M\}$. For a matrix $\mathbf{D}$, we denote by $[\mathbf{D}]_{ij}$ the $(i,j)$th entry of $\mathbf{D}$. The set of all natural numbers is denoted by  $\mathbb{N}$  and $\mathbb{N}_{0}$ denotes the set of all non-negative integers. Futher, $\mathbf{I}_m$ denotes the identity matrix of size $m \times m$ and $\boldsymbol{0}_{m}$ denotes a zero matrix of size $m \times m$.
\section{Preliminaries}
\label{preliminaries}
In this section, we review some definitions and basic concepts related to  sum-rank metric codes and FCCs.
\subsection{Sum-rank metric codes} 
  \begin{defn}[Sum-rank weight{\cite{b7}}]
 For $q,$  a prime power and  $m \ge n$ being positive integers, let $\mathbb{F}_q^{m \times n}$ be the set of all $m \times n$ matrices with entries from the finite field $\mathbb{F}_q$.  For a matrix $\mathbf{M} \in \mathbb{F}_q^{m \times n}$, we denote its rank by $\mathrm{rk}(\mathbf{M})$, and let \[\mathbb{M} = \mathbb{F}_q^{m_1 \times n_1} \times \cdots \times \mathbb{F}_q^{m_\ell \times n_\ell}\] be the $\mathbb{F}_q$-linear vector space obtained as the cartesian product of matrix spaces, where $\ell, m_1, \ldots, m_\ell, n_1, \ldots, n_\ell$ are positive integers with  \[m_1 \ge \cdots \ge m_\ell \quad \text{and} \quad n_i \le m_i \ \text{for all } i \in [\ell].\]Moreover, $n = n_1 + \cdots + n_\ell.$ If $m_1 = \cdots = m_\ell$, we simply write $m$ in place of $m_i$.

  For an element $\mathbf{c}\in\mathbb{M}$ written as $\mathbf{c} = (C_1, \ldots , C_\ell)$ (a vector of matrices), with $C_i\in\mathbb{F}_q^{m_i\times n_i}$ for $i\in[\ell]$, its sum-rank weight is given by
  \begin{equation*}
	\mathrm{srk}(\mathbf{c}) = \sum_{i=1}^\ell rk(C_i).
    \end{equation*}
    \end{defn}
  
  \begin{defn}[\cite{b7}]
  Let $d_{srk}$ be the map $$\begin{array}{ccccc}
  	d_{srk}  &: & \mathbb{M} \times \mathbb{M }& \longrightarrow & \mathbb{N} \\
  	& & (\mathbf{c}, \mathbf{d}) & \longrightarrow & \mathrm{srk}(\mathbf{c}-\mathbf{d}),
  \end{array}$$ Then $(\mathbb{M}, d)$ is a metric space.
  \end{defn}
  \begin{defn}[Sum-rank distance{\cite{b7}}] 
   Let $\mathcal{C} \subseteq \mathbb{M}$ be a sum-rank metric code. The minimum sum-rank distance of $\mathcal{C}$ is defined as\[d_{\mathrm{srk}}(\mathcal{C})\;=\;\min_{\substack{\mathbf{x}, \mathbf{y} \in \mathcal{C} \\ \mathbf{x} \neq \mathbf{y}}} \mathrm{srk}(\mathbf{x} - \mathbf{y}).\]\end{defn}
   \label{rem1}
   The sum-rank metric is a natural generalization of both the Hamming metric and the rank-metric. If $\ell = 1$, then $\mathbb{M} = \mathbb{F}_q^{m \times n}$ and the sum-rank metric coincides with the rank-metric on $\mathbb{M}$. In this case, a sum-rank metric code $\mathcal{C}$ is a rank-metric code. On the other hand, if $m_1 = \cdots = m_\ell = 1$, then $\mathbb{M} = \mathbb{F}_q^{n}$, and the sum-rank metric coincides with the Hamming metric.

   Throughout, we consider FCC for sum-rank metric codes over $(\mathbb{F}_q^{m \times m})^k$. The following example illustrates the sum-rank distance between two codewords of a sum-rank metric code.
   \begin{example}
   \label{example1}
	For $q = 2, m = 2,$ and   $k = 3$, let $\mathcal{C} \subseteq (\mathbb{F}_2^{2 \times 2})^3$ be a sum-rank metric code. Let $\mathbf{c}_1,\mathbf{c}_2 \in \mathcal{C}$ be two codewords given by  \[\mathbf{c}_1 =\left[\begin{pmatrix}
		0 & 0\\
		0 & 0
	\end{pmatrix},
	\begin{pmatrix}
		1 & 1\\
		0 & 1
	\end{pmatrix},
	\begin{pmatrix}
		1 & 0\\
		0 & 1
	\end{pmatrix}\right]\] 
	\[\mathbf{c}_2 =\left[
	\begin{pmatrix}
		1 & 0\\
		0 & 1
	\end{pmatrix},
	\begin{pmatrix}
		0 & 1\\
		1 & 0
	\end{pmatrix},
	\begin{pmatrix}
		1 & 0\\
		1 & 0
	\end{pmatrix}
	\right]\]\[
	\mathbf{c}_1-\mathbf{c}_2
	=
	\left[
	\begin{pmatrix}
		1 & 0\\
		0 & 1
	\end{pmatrix},
	\;
	\begin{pmatrix}
		1 & 0\\
		1 & 1
	\end{pmatrix},
	\;
	\begin{pmatrix}
		0 & 0\\
		1 & 1
	\end{pmatrix}
	\right]
	\]
	The sum-rank distance between $\mathbf{c}_1$ and $\mathbf{c}_2$ is
	\[
	d_{\mathrm{srk}}(\mathbf{c}_1,\mathbf{c}_2)
	=
	\sum_{i=1}^{3}
	\mathrm{rk}\!\left(
	\mathbf{c}_1^{(i)} - \mathbf{c}_2^{(i)}
	\right) = 2 + 2 + 1 = 5.
	\]
    \end{example}
\begin{defn}[\cite{b5}]
\label{MRD}
Let $q$ be a prime power and $m \ge 2$. An $[m,1,m]$ constant maximum rank-metric code over $\mathbb{F}_{q^m}$ is a linear code $\mathcal{C} \subseteq \mathbb{F}_{q^m}^{\,m}$ of dimension $1$ over $\mathbb{F}_{q^m}$ such that $d_{\mathrm{rk}}(\mathbf{X},\mathbf{Y}) = m$ for all distinct  $\mathbf{X},\mathbf{Y} \in \mathcal{C},$ where $\mathbf{X}, \mathbf{Y}$ are matrices of size $m \times m.$
\end{defn}
\begin{example}
	\label{example2}
	Let $q = 2$ and $\mathbb{F}_{4} = \{0, 1, \alpha, \alpha^2\}$, where $\alpha$ is a root of the irreducible polynomial $f(\lambda) = \lambda^2 + \lambda + 1$. Consider the one-dimensional $\mathbb{F}_{2^m}$-linear code of length $m = 2$  generated by generator matrix $\mathbf{G} = [1 \; \alpha]$ with information vectors as $\mathbf{u} = (u), u \in \mathbb{F}_4$. Code vectors are given by $\mathbf{v} = \mathbf{uG}$.\[\mathcal{C}= \{\mathbf{v}\mid\mathbf{v}=(u,u\alpha), u \in \mathbb{F}_4\}.\] Define a vector representation $\theta^{-1} : \mathbb{F}_4 \to \mathbb{F}_2^2$ by \[1 \leftrightarrow (1,0)^T, \qquad \alpha \leftrightarrow (0,1)^T,\] which implies $\alpha^2 = \alpha + 1 \leftrightarrow (1,1)^T.$ Applying $\theta^{-1}$ column-wise to each codeword yields the associated $2 \times 2 $ matrix code \[
	\mathcal{M} =
	\left\{
	\begin{pmatrix}
		0 & 0 \\
		0 & 0
	\end{pmatrix},
	\begin{pmatrix}
		1 & 0 \\
		0 & 1
	\end{pmatrix},
	\begin{pmatrix}
		0 & 1 \\
		1 & 1
	\end{pmatrix},
	\begin{pmatrix}
		1 & 1 \\
		1 & 0
	\end{pmatrix}
	\right\}.
	\]
	This matrix code is the $\mathbb{F}_2$-representation of the $[2,1,2]$ $\mathbb{F}_{4}$-linear code $\mathcal{C}$. The rank distance between any two codewords is 2.
\end{example}
In \cite{b5}, an explicit way of construction of such $[m,1,m]$ 1-dimensional rank metric codes has been presented. . 
\subsection{Function-Correcting Codes \cite{b1}}
In \cite{b1}, the concept of function-correcting codes is introduced, along with a coding framework that formalizes the basic definitions of function-correcting codes and irregular distance codes over the binary field $\mathbb{F}_2$. In \cite{b1}, $\mathbb{Z}_2$ is used for $\mathbb{F}_2.$
\begin{defn}[\cite{b1}]
An encoding function $\mathrm{Enc} : \mathbb{Z}_2^k \rightarrow \mathbb{Z}_2^{k+r}$, with $\mathrm{Enc}(\mathbf{u}) = (\mathbf{u}, \mathbf{p}(\mathbf{u})), \quad \mathbf{u} \in \mathbb{Z}_2^k,\quad \mathbf{p}(\mathbf{u}) \in \mathbb{Z}_2^{r}$ defines a function-correcting code for the function $f : \mathbb{Z}_2^k \rightarrow \mathrm{Im}(f)$ if, for all $\mathbf{u}_1, \mathbf{u}_2 \in \mathbb{Z}_2^k$ such that $f(\mathbf{u}_1) \neq f(\mathbf{u}_2)$, it holds that $d\big(\mathrm{Enc}(\mathbf{u}_1), \mathrm{Enc}(\mathbf{u}_2)\big) \geq 2t + 1.$
\end{defn}
The check vector $\mathbf{p}(\mathbf{u})$ is the redundancy vector, and $r$ denotes the redundancy. The lower and upper bounds on optimal redundancy are established in \cite{b1} using connection between FCCs and irregular distance codes for the Hamming metric.  The distance requirement matrix and irregular distance codes are defined as follows.
\begin{defn}[Distance Requirement Matrix ~\cite{b1}]
Let $\mathbf{u}_1, \ldots, \mathbf{u}_M \in \mathbb{Z}_2^k$. The \emph{distance requirement matrix}
$\mathbf{D}_f(t, \mathbf{u}_1, \ldots, \mathbf{u}_M)$ of a function $f$ is defined as the $M \times M$
matrix with entries
\[
	[\mathbf{D}_f(t,\! \mathbf{u}_1, \dots, \mathbf{u}_M)]_{ij}\!=\!
	\begin{cases}
		\begin{aligned}
			&[2t\!+\!1\!-\!d(\mathbf{u}_i,\mathbf{u}_j)]^+\!, \text{if } f(\!\mathbf{u}_i\!) \!\neq\! f(\!\mathbf{u}_j\!), \\
			&0, \quad \text{otherwise}.
		\end{aligned}
	\end{cases}
	\]
Let $\mathcal{P} = \{\mathbf{p}_1, \mathbf{p}_2, \ldots, \mathbf{p}_M\} \subseteq \mathbb{Z}_2^{r}$ be a code of length $r$ and cardinality $M$.
\end{defn}
\begin{defn}[$\mathbf{D}$-Code ~\cite{b1}]
Let $\mathbf{D} \in \mathbb{N}_0^{M \times M}$. A code $\mathcal{P} = \{\mathbf{p}_1, \mathbf{p}_2, \ldots, \mathbf{p}_M\}$ is called a $\mathbf{D}$-code if there exists an ordering of the codewords of $\mathcal{P}$ such that $d(\mathbf{p}_i, \mathbf{p}_j) \geq [\mathbf{D}]_{ij}, \quad \forall\, i,j \in [M].$ Furthermore, we define $N(\mathbf{D})$ as the smallest integer $r$ such that there exists a $\mathbf{D}$-code of length $r$. If $[\mathbf{D}]_{ij} = D$ for all $i \neq j$, we write $N(M,D)$.
\end{defn}
A $\mathbf{D}$-code requires that the distance between each pair of codewords is at least the value specified by the corresponding entry in the distance requirement matrix. Motivated by this framework for the Hamming metric, we extend this framework  to the sum-rank metric in the next section and obtain our results.
\section{General  Results on the Optimal Redundancy}
\label{General  Results on the Optimal Redundancy}
In this section, we first establish the results on the optimal redundancy of Function-Correcting Sum-Rank Codes (FCSRCs). We then propose a Plotkin-like bound for irregular-sum-rank distance codes. 
\subsection{A Connection between FCSRCs and Irregular-Sum-Rank-Distance Codes}
\label{A Connection between FCSRCs and Irregular-Sum-Rank-Distance Codes}
Let $ \mathbf{u} \in (\mathbb{F}_q^{m \times m})^k.$ be the message and let $ f : (\mathbb{F}_q^{m \times m})^k \to Im(f) = \{ f(\mathbf{u}) \mid \mathbf{u} \in (\mathbb{F}_q^{m \times m})^k \} $ be a function computed on $ \mathbf{u}$ with expressiveness $ E = |Im(f)| $. The message is encoded via the encoding function $\mathrm{Enc} : (\mathbb{F}_q^{m \times m})^k \to (\mathbb{F}_q^{m \times m})^{k+r}$,  where $\mathrm{Enc}(\mathbf{u}) = (\mathbf{u}, \mathbf{p}(\mathbf{u})),$ with $ \mathbf{p}(\mathbf{u}) \in (\mathbb{F}_q^{m \times m})^r$ being the redundancy vector and $r$ the redundancy. The resulting codeword $ \mathrm{Enc}(\mathbf{u})$ is transmitted over a channel, resulting in $ \mathbf{y} \in (\mathbb{F}_q^{m \times m})^{k+r} $ with $ d_{srk}(\mathrm{Enc}(\mathbf{u}), \mathbf{y}) \leq t $. The formal definition of FCSRC is given below.
\begin{defn}[Function-Correcting Sum-Rank Codes]
\label{Function-Correcting Sum-Rank Codes}
    An encoding function \( \mathrm{Enc} : (\mathbb{F}_q^{m \times m})^k \to (\mathbb{F}_q^{m \times m})^{k+r} \),  
	$ \mathrm{Enc}(\mathbf{u}) = (\mathbf{u}, \mathbf{p}(\mathbf{u})) $ defines a function-correcting sum-rank code (FCSRC) for the function $ f : (\mathbb{F}_q^{m \times m})^k \to Im(f) $  if for all $ \mathbf{u}_1, \mathbf{u}_2 \in (\mathbb{F}_q^{m \times m})^k $ with $ f(\mathbf{u}_1) \neq f(\mathbf{u}_2)$, we have $d_{srk}(\mathrm{Enc}(\mathbf{u}_1), \mathrm{Enc}(\mathbf{u}_2)) \geq 2t + 1$. 
\end{defn}
We now introduce the concept of optimal redundancy of an FCSRC associated with a function $f$, which is the key performance measure in FCCs for any metric.

\begin{defn}[Optimal Redundancy]
\label{Optimal Redundancy}
The optimal redundancy \( r_{srk}^f(m, k, t) \) is defined as the smallest integer \( r \) for which there exists an FCSRC with an encoding function \( \mathrm{Enc} : (\mathbb{F}_q^{m \times m})^k \to (\mathbb{F}_q^{m \times m})^{k+r} \) that enables recovery of \( f(\mathbf{u}) \) under $t$ sum-rank errors.
\end{defn}
The definition of sum-rank distance requirement  matrix (DRM) and $\mathbf{D}$-irregular-sum-rank-distance code associated with a function \(f\)  follows.

\begin{defn}[DRM] 
\label{DRM}
Let $\mathbf{u}_1, \dots, \mathbf{u}_M \in (\mathbb{F}_q^{m \times m})^k$. The DRM 	$\mathbf{D}_{srk}^f(t, \mathbf{u}_1, \dots, \mathbf{u}_M)$ of a function $f$ is defined as the $M \times M$ matrix with entries
	\[
	[\!\mathbf{D}_{srk}^f\!(t,\! \mathbf{u}_1, \dots, \mathbf{u}_M)\!]_{ij}\!=\!
	\begin{cases}
		\begin{aligned}
			&\![\!2t\! +\!1\!-\!d_{srk}\!(\!\mathbf{u}_i,\! \mathbf{u}_j\!)\!]^+\!, \text{if }\! f(\!\mathbf{u}_i\!)\!\neq\!f(\!\mathbf{u}_j\!), \\
			&0, \quad \text{otherwise}.
		\end{aligned}
	\end{cases}
	\]
\end{defn}
\begin{defn}[$\mathbf{D}_{srk}$-code]
\label{D-code}
Let $\mathbf{D} \in \mathbb{N}_0^{M \times M}$ and $\mathcal{P} = \{ \mathbf{p}_1, \mathbf{p}_2, \ldots, \mathbf{p}_M \} \subseteq (\mathbb{F}_q^{m \times m})^k$ be a code of length $r$ and cardinality $M$. Then, $\mathcal{P} = \{ \mathbf{p}_1, \mathbf{p}_2, \ldots, \mathbf{p}_M \}$ is a $\mathbf{D}$-irregular-sum-rank-distance code ($\mathbf{D}_{srk}$-code), if there exists an ordering of the codewords of $\mathcal{P}$ such that $d_{srk}(\mathbf{p}_i, \mathbf{p}_j) \geq [\mathbf{D}]_{ij}$ for all $i, j \in [M]$.
\end{defn}
The smallest integer $r$ such that there exists a $\mathbf{D}_{srk}$-code of length $r$ is denoted by $N_{srk}(\mathbf{D})$. With  these definitions, a link between FCSRCs and irregular-sum-rank-distance codes is established in the following theorem the proof of which is along the same lines as that was given in \cite{b1} for the Hamming metric. 

\begin{thm}
\label{theorem1}
For any function $f : (\mathbb{F}_q^{m \times m})^k \to \mathrm{Im}(f)$, the optimal redundancy satisfies
\[r_{srk}^f(m,k,t)\;=\;N_{srk}(\mathbf{D}_{srk}^f(t; \mathbf{u}_1, \ldots, \mathbf{u}_{(q^{m \times m})^{k}})),\]where $\{\mathbf{u}_1, \ldots, \mathbf{u}_{(q^{m \times m})^{k}}\} \in (\mathbb{F}_q^{m \times m})^k$ denotes the set of all $k$ length vectors of matrices over $\mathbb{F}_q^{m \times m}$.
\end{thm}
\subsection{Simplified Bounds on Optimal Redundancy}
\label{Simplified Bounds on Optimal Redundancy}
In this subsection, we first compute simplified lower bounds on the optimal redundancy of FCSRCs. Using an arbitrary subset of information vectors $\{\mathbf{u}_1, \dots, \mathbf{u}_{M}\} \in (\mathbb{F}_q^{m \times m})^k$, we can obtain lower bounds as follows:
\begin{cor}
\label{cor1}
Let $\mathbf{u}_1, \ldots, \mathbf{u}_M \in (\mathbb{F}_q^{m \times m})^k$ be an arbitrary subset of distinct vectors. Then the optimal redundancy of a FCSRC satisfies
\[
r_{srk}^f(m, k, t)\;\ge\; N_{srk}\!(\mathbf{D}_{srk}^f(t; \mathbf{u}_1, \ldots, \mathbf{u}_M)).
\]
For any function $f$ with $|\mathrm{Im}(f)| \ge 2$, the redundancy satisfies
\[
r_{srk}^f(m,k,t) \;\ge\; N_{srk}(2, 2t)
\;=\;
\left\lceil \frac{2t}{m} \right\rceil .
\]
\end{cor}
We derive bounds on the optimal redundancy of FCSRCs with respect to funciton distance matrix which depends on sum-rank distance between two distinct function values defined as follows. 
\begin{defn}[Function Distance]
\label{Function Distance}
The minimum sum-rank distance between any pair of information vectors that evaluate to $f_1$ and $f_2$, is 
	\[
	d_{srk}^f(f_1, f_2) \triangleq \! \min_{\mathbf{u}_1, \mathbf{u}_2 \in (\mathbb{F}_q^{m \times m})^k} \! d_{srk}(\mathbf{u}_1, \mathbf{u}_2) \; \]
    \[\text{s.t.} \; f(\mathbf{u}_1) = f_1, f(\mathbf{u}_2) = f_2.
	\]
\end{defn}
Based on this, the definition of the sum-rank function-distance matrix (FDM) of $f$ follows.
\begin{defn}[Function Distance Matrix] 
\label{FDM}
    The function distance matrix of $f$ with $|Im(f)|$ = E is defined with $E \times E$ matrix $\mathbf{D}_{srk}^f(t, f_1, \ldots, f_E)$ with entries $[ \mathbf{D}_{srk}^f(t, f_1, \ldots, f_E)]_{ij} = [2t + 1 - d_{srk}^f(f_i, f_j)]^+, \; \text{if } i \neq j$ and $[ \mathbf{D}_{srk}^f(t, f_1, \ldots, f_E)]_{ii} = 0$, for $t$ sum-rank error correction.
\end{defn}
\begin{thm}
\label{theorem2}
For any arbitrary function $f : (\mathbb{F}_q^{m \times m})^k \rightarrow \mathrm{Im}(f)$, we have 
	$
	r_{srk}^f(q, k, t) \leq N_{srk}(\mathbf{D}_{srk}^f(t, f_1, \ldots, f_E)) 
	$.
\end{thm}
The proof of this theorem is also along the same lines as the corresponding theorem for the Hamming metric given in \cite{b1}. 
\subsection{Plotkin-like Bound for $N_{srk}(D)$}
\label{Plotkin Bound}

 In this subsection, we propose a Plotkin-like bound for irregular-sum-rank-distance codes. 
\begin{thm}
\label{theorem3}
For any distance matrix $\mathbf{D} \in \mathbb{N}_0^{M \times M}$ and for irregular-sum-rank-distance codes over
$\mathbb{F}_q$, we have
\[N_{srk}(\mathbf{D})\;\ge\;\frac{2q^m}{m(M^{2}(q^m-1)-a(q^m-a))}\sum_{1 \le i < j \le M} [\mathbf{D}]_{i,j},\] where $a = M \mod {q^m}$.
\end{thm}
\begin{IEEEproof}
	Let $N_{srk}(\mathbf{D}) = r$ and $\left\{ \mathbf{p}_i \right\}_{i=1}^{M}$ be codewords of a $\mathbf{D}_{srk}$-code of length $r$. Since $\left\{ \mathbf{p}_i \right\}_{i=1}^{M}$ form a $\mathbf{D}_{srk}$-code, by definition we have $[\mathbf{D}]_{ij} \le d_{srk}(\mathbf{p}_i, \mathbf{p}_j)  \quad \forall i, j$. Therefore, \begin{equation}
		\sum_{i,j: i<j} [\mathbf{D}]_{ij} \le \sum_{i,j: i<j} d_{srk}(\mathbf{p}_i, \mathbf{p}_j). 
		\label{eqn1}
	\end{equation}
    We have  \[\sum_{1 \le i < j \le M} d_{\mathrm{srk}}(\mathbf{p}_i, \mathbf{p}_j)\;=\;\sum_{\ell=1}^{r} \sum_{1 \le i < j \le M} \mathrm{rk}\!\left(\mathbf{p}_i^{(\ell)} - \mathbf{p}_j^{(\ell)}\right).\]
    Consider the $\ell$-th coordinate of all codewords. The contribution of this coordinate to the total pairwise sum-rank distance is maximized when the matrices at this coordinate are chosen to have maximum possible pairwise rank distance.\\
    \indent From Definition ~\ref{MRD}, among all the subsets of $\mathbb{F}_q^{m \times m}$, the rank-metric code $\mathcal{M}$ with |$\mathcal{M}$| = $q^m$ is chosen such that every pair of distinct matrices differs by full rank $m$.\\
    \indent Let $a = M \bmod q^m$. To distribute them evenly over $M$ locations, $a$ matrices appear $\lceil\frac{M}{q^m}\rceil$ times, while the remaining $(q^m-a)$ matrices appear $\lfloor\frac{M}{q^m}\rfloor$ times. \\
    \indent Each of the $a$ matrices that appear $\lceil\frac{M}{q^m}\rceil$ times contributes $\lceil \frac{M}{q^m}\rceil \left(M - \lceil \frac{M}{q^m}\rceil\right)$ ordered pairs. Similarly, each of the $q^m - a$ matrices that appear $\lfloor \frac{M}{q^m} \rfloor$ times contributes $\lfloor \frac{M}{q^m}\rfloor \left(M - \lfloor \frac{M}{q^m}\rfloor\right)$ ordered pairs. Since each such pair contributes the maximum rank distance $m$, summing these two cases and dividing by 2 (since we consider pairs with $i<j$) gives the upper bound.

    \begin{equation}
    \begin{aligned}
    \sum_{1 \le i < j \le M}\!rk(\mathbf{p}_i^{(\ell)} - \mathbf{p}_j^{(\ell)})\le&\frac{m}{2}\Bigg[ a \left\lceil \frac{M}{q^m}\right\rceil\left( M - \left\lceil \frac{M}{q^m} \right\rceil\right) \\&+(q^m-a)\left\lfloor \frac{M}{q^m} \right\rfloor\left( M - \left\lfloor \frac{M}{q^m} \right\rfloor \right)\!\Bigg]
\end{aligned}
\label{eqn2}
\end{equation}
    \begin{equation}
    \begin{aligned}
    \sum_{1 \le i < j \le M} d_{srk}(\mathbf{p}_i,\mathbf{p}_j)\le&\frac{rm}{2}\Bigg[a \left\lceil \frac{M}{q^m} \right\rceil\left( M - \left\lceil \frac{M}{q^m} \right\rceil \right) \\&\quad+(q^m-a)\left\lfloor \frac{M}{q^m} \right\rfloor\left( M - \left\lfloor \frac{M}{q^m} \right\rfloor \right)\!\Bigg]
\end{aligned}
\label{eqn3}
\end{equation}
    From ~\eqref{eqn1} and ~\eqref{eqn3}
    \begin{equation}\sum_{1 \le i < j \le M} [\mathbf{D}]_{i,j}\;\le\;\frac{rm}{2q^m}\bigl( M^{2}(q^m-1) - a(q^m-a) \bigr).
    \label{eqn4}
     \end{equation}
    Thus we have the desired bound as,
    \[N_{srk}(\mathbf{D})\;\ge\;\frac{2q^m}{m(M^{2}(q^m-1)-a(q^m-a))}\sum_{1 \le i < j \le M} [\mathbf{D}]_{i,j},\] where $a = M \mod{q^m}$.
    
\end{IEEEproof}
\begin{rem}
\label{rem2}
	For $m=1$, we have $a = M\quad mod \quad q$. This bound reduces to the generalized Plotkin bound for Hamming metric over $\mathbb{F}_q$:
	\[N(\mathbf{D}) \geq \frac{2q}{M^2 (q-1) - a(q-a) }\sum_{1 \leq i < j \leq M} [\mathbf{D}]_{i,j}\].
\end{rem}
\begin{rem}
\label{rem3}
	For regular-distance codes with minimum sum-rank distance $d_{srk},\sum_{i,j: i<j} [\mathbf{D}]_{ij} \geq \frac{M(M-1)}{2}d_{srk}$. Then, the Plotkin bound in Theorem ~\ref{theorem3} reduces to \[N_{srk}(M,D) \ge \frac{q^{m}(M-1)}{m\,(q^{m}-1)\,M}\, d_{srk}\]. This is same as the Plotkin bound stated in \cite{b8} for sum-rank metric codes as upper bound for the size of the code $M$.
\end{rem}
\begin{rem}
\label{rem4}
    For a rank-metric code $\mathcal{C} \subseteq \mathbb{F}_{q^{m}}^{m}$ with minimum rank distance $d$, if length of code, $r = 1$, it reduces to the Plotkin bound for rank-metric which is stated as \[|\mathcal{C}| \le \frac{q^md}{q^md - (q^{m}-1)m},\, d > \frac{q^{m}-1}{q^{m}}m\].
\end{rem}
\begin{example}
\label{example3}
	Let $q = 2, m = 2,$ and  $M = 15$, Example~\ref{example2} gives $[2,1,2]$ 1-dimensional rank-metric code $\mathcal{M} = \left\{
	\begin{pmatrix}
		0 & 0 \\
		0 & 0
	\end{pmatrix},
	\begin{pmatrix}
		1 & 0 \\
		0 & 1
	\end{pmatrix},
	\begin{pmatrix}
		0 & 1 \\
		1 & 1
	\end{pmatrix},
	\begin{pmatrix}
		1 & 1 \\
		1 & 0
	\end{pmatrix}
	\right\} \in \mathbb{F}_2^{2 \times 2}.$ The rank distance between any two codewords in $\mathcal{M}$ is 2. Taking $a = M \mod q^m = 15 \mod 4 = 3$ of them for $\lceil\frac{M}{q^m}\rceil = \lceil\frac{15}{4}\rceil = 4$ times and $q^m - a = 1$ of them for $\lfloor\frac{M}{q^m}\rfloor = \lfloor\frac{15}{4}\rfloor = 3$ times contributes maximum, that is 168, to the $\sum_{1 \le i < j \le M} rk(\mathbf{p}_i^{(\ell)} - \mathbf{p}_j^{(\ell)})$, total pairwise rank distance for a single column $l$, as each pairwise rank distance gives full rank 2, this is given by ~\eqref{eqn2} .
\end{example}
\subsection{Plotkin Bound for FCSRCs for linear functions}
\label{linear functions}
A function \( f : (\mathbb{F}_q^{m \times m})^k \to (\mathbb{F}_q^{m \times m})^l \) is said to be linear if it satisfies the following condition:
\[
f(\alpha \mathbf{x} + \beta \mathbf{y}) = \alpha f(\mathbf{x}) + \beta f(\mathbf{y}), \quad \forall\, \mathbf{x}, \mathbf{y} \in (\mathbb{F}_q^{m \times m})^k,\ \alpha,\beta \in \mathbb{F}_q .\]It can be expressed as a matrix operation $f(\mathbf{x}) = \mathbf{Fx}$, for some $\mathbf{F} \in (\mathbb{F}_q^{m \times m})^{l \times k}$. The kernel of $f$, or the null space of $\mathbf{F}$, is denoted by $\ker(f)$. We will consider linear functions for this subsection such that $l \le k$ and $\mathbf{F}$ is full rank. The following theorem proposes Plotkin bound for FCSRCs for linear functions.
\begin{thm}
\label{theorem4}
    For a linear function $f : (\mathbb{F}_q^{m \times m})^k \to (\mathbb{F}_q^{m \times m})^l$, the optimal redundancy of an $(f,t)$-FCSRC satisfies 
    \begin{align*}
    r_{\mathrm{srk}}^f(m,k,t)\ge\Bigl(\frac{q^m}{m(q^m-1)}\Bigr)(2t+1)(1-q^{-m^2l})-k \\+\frac{s_{\mathrm{srk}}}{m(1-q^{-m})q^{m^2k}},
    \end{align*}
    where $s_{srk}=\sum_{\mathbf{x}\in \ker(f)} w_{srk}(\mathbf{x})$, i.e., the sum of sum-rank weights of the vectors in $\ker(f)$.
\end{thm}
\begin{IEEEproof}
    Let $M=q^{km^2}$ be the total number of sum-rank vectors in $(\mathbb{F}_q^{m \times m})^k$ and $\mathbf{D}$ denote the distance requirement matrix of size $M\times M$ defined by
    \[[\mathbf{D}]_{ij}=\begin{cases}\left[2t+1-d_{\mathrm{srk}}(\mathbf{u}_i,\mathbf{u}_j)\right]^+,& \text{if } f(\mathbf{u}_i)\neq f(\mathbf{u}_j),\\0, & \text{otherwise}.\end{cases}\] 
    Let $\mathcal{P} = \{\mathbf{p}_1,\mathbf{p}_2,\mathbf{p}_3,...,\mathbf{p}_{M}\}$ be the set of parity vectors such that $\operatorname{Enc}(\mathbf{u}_i)=(\mathbf{u}_i,\mathbf{p}_i)$ forms an $(f,t)$-FCSRC. By definition, we have $d_{srk}(\mathbf{p}_i,\mathbf{p}_j)\ge [\mathbf{D}]_{ij}$ for all $i,j$. 
    For any $\mathbf{D}_{\mathrm{srk}}$-code of length $r$, Theorem ~\ref{theorem3} gives \[\sum_{1\le i<j\le M}d_{srk}(\mathbf{p}_i,\mathbf{p}_j)\le\frac{rm}{2q^m}\Bigl[M^2(q^m-1)-a(q^m-a)\Bigr],\]where $a=M \bmod q^m$. 
    Since $M=q^{m^2k}$, we have $a=0$. Moreover $d_{\mathrm{srk}}(\mathbf{p}_i,\mathbf{p}_j)\geq [\mathbf{D}]_{ij}$, so \[\sum_{1\le i<j\le M}[\mathbf{D}]_{ij}\le\frac{rm}{2q^m}\cdot M^2(q^m-1).\]
    Summing over all i and j, we get 
    \begin{equation}\sum_{i,j}[\mathbf{D}]_{ij}\le\frac{rm}{q^m}\cdot M^2(q^m-1).
    \label{eqn5}
    \end{equation}
    Here $d_{srk}(\mathbf{u}_i,\mathbf{u}_j) = w_{srk}(\mathbf{u}_i-\mathbf{u}_j) = \sum_{i=1}^k{rk(\mathbf{u}_i^{(k)} - \mathbf{u}_j^{(k)})}$. $ker(f) = \{\mathbf{u} : f(\mathbf{u}) =0, \mathbf{u} \in (\mathbb{F}_q^{m \times m})^k\}, \; |ker(f)| = q^{m^2(k-l)}$. Each coset of kernel contains $q^{m^2(k-l)}$ number of vectors. Let $\mathbf{u}_1$ be a zero vector. The entry $[\mathbf{D}]_{ij}$ of DRM is zero when $f(\mathbf{u}_i) = f(\mathbf{u}_j)$. For first column with $\mathbf{u}_1 = \mathbf{0}$, there will be $q^{m^2(k-l)}$ number of zero entries when $\mathbf{u}_i \in ker(f)\; \forall i \in (\mathbb{F}_q^{m \times m})^k$.\\
    When two vectors $\mathbf{u}_i$ and $\mathbf{u}_j$ belong to same coset of kernel, then $f(\mathbf{u}_i) = f(\mathbf{u}_j)$ which results in zero entry in a column. Thus, every column in $\mathbf{D}$ has $q^{m^2(k-l)}$ zero entries.\\
    Let $\mathbb{F}_q^{m\times m}$ denote the set of all $m\times m$ matrices over $\mathbb{F}_q$. Consider vectors over matrices of length $k$. Let any two columns be indexed by $\mathbf{u}_j$ and $\mathbf{u}_{j'}$, Let $\mathbf{u}_i$ be the vector such that $\mathbf{u}_i \notin \mathbf{u}_j + ker(f)$. This results a nonzero entry in column j given by \[2t+1-d_{srk}(\mathbf{u}_i,\mathbf{u}_j) = 2t+1 - w_{srk}(\mathbf{u}_i-\mathbf{u}_j).\]
    Let $ \mathbf{v} = \mathbf{u}_{j'} - \mathbf{u}_j$. By additive closure and existence of additive inverses in $(\mathbb{F}_q^{m \times m})^k$, the vector $\mathbf{v}=\mathbf{u}_{j'}-\mathbf{u}_j$ belongs to $(\mathbf{F}_q^{m \times m})^k$. Now define a new row index $i'$ such that $\mathbf{u}_{i'} = \mathbf{u}_i + \mathbf{v}$.
    \[\mathbf{u}_{i'}-\mathbf{u}_{j'}=(\mathbf{u}_i+\mathbf{v})-\mathbf{u}_{j'}=\mathbf{u}_i+(\mathbf{u}_{j'}-\mathbf{u}_j)-\mathbf{u}_{j'}=\mathbf{u}_i-\mathbf{u}_j.\]
    Thus, $d_{srk}(\mathbf{u}_i,\mathbf{u}_j) = d_{srk}(\mathbf{u}_{i'},\mathbf{u}_{j'})$. Hence, $[\mathbf{D}]_{i'j'}=2t+1-d_{srk}(\mathbf{u}_{i'}, \mathbf{u}_{j'})=2t+1-d_{srk}(\mathbf{u}_i,\mathbf{u}_j)=[\mathbf{D}]_{ij}$. The nonzero entries of the two columns are identical, differing only in their ordering. Therefore, the columns are permutations of each other. Thus, the sum of non-zero entries of each column of $\mathbf{D}$ will be the same. Hence, \[\sum_{i,j}[\mathbf{D}]_{ij}= \text{(no. of columns)} \times \text{(sum of one column)}.\] To compute sum of one column, consider first column with $\mathbf{u}_1 = \mathbf{0}$. \[\sum_{i,j}[\mathbf{D}]_{ij} = M \cdot S_1,\] where $S_1=\sum_{i=1}^M [\mathbf{D}]_{i1}$ and $[\mathbf{D}]_{i1} \geq 2t+1 - d_{srk}(\mathbf{u}_i,0)$. For $\mathbf{u}_1=0$, we have $d_{srk}(\mathbf{u}_i,0)=w_{srk}(\mathbf{u}_i)$. If $\mathbf{u}_i\in\ker(f)$, then $f(\mathbf{u}_i)=f(0)$ and hence $[D]_{i1}=0$. If $\mathbf{u}_i\notin\ker(f)$, then\[[\mathbf{D}]_{i1}\ge 2t+1-w_{srk}(\mathbf{u}_i).\] Therefore, $S_1\ge\sum_{\mathbf{u}\notin\ker(f)}\Bigl(2t+1-w_{srk}(\mathbf{u})\Bigr).$ 
    \[\sum_{\mathbf{u}\notin\ker(f)}(2t+1) = 2t+1\Bigl(q^{m^2k}-q^{m^2(k-l)}\Bigr)\]
    and $\sum_{\mathbf{u}\notin\ker(f)}w_{srk}(\mathbf{u}) = \sum_{\mathbf{u}}w_{srk}(\mathbf{u}) - \sum_{\mathbf{u}\in\ker(f)}w_{srk}(\mathbf{u})$. Let $s_{srk}=\sum_{\mathbf{u}\in\ker(f)} w_{srk}(\mathbf{u})$.
   Thus,\begin{equation}
   \begin{aligned}
    S_1 \ge 2t+1\Bigl(q^{m^2k}-q^{m^2(k-l)}\Bigr)-\Bigl(\sum_{\mathbf{u}}w_{srk}(\mathbf{u})-s_{srk}\Bigr)
    \\
    =2t+1\Bigl(q^{m^2k}-q^{m^2(k-l)}\Bigr) - \sum_{\mathbf{u}}w_{srk}(\mathbf{u}) + s_{srk}.
   \end{aligned}
   \label{eqn6}
   \end{equation}  
   Consider the $\mathbb{F}_q$-linear isomorphism \[\Phi:(\mathbb{F}_q^{m\times m})^k\to\mathbb{F}_{q^m}^{km}\]obtained by mapping each matrix column-wise to an element of $\mathbb{F}_{q^m}$ using a fixed basis. For any any $\mathbf{u} \in (\mathbb{F}_q^{m \times m})^k$, $w_{srk}(\mathbf{u})\le w_H(\Phi(\mathbf{u}))$. Summing over all $u$ and using the bijectivity of $\Phi$, \[\sum_{\mathbf{u}} w_{srk}(\mathbf{u})\le\sum_{\mathbf{y}\in\mathbb{F}_{q^m}^{km}} w_H(\mathbf{y}).\] The sum of Hamming weights of all vectors of length $km$ over $\mathbb{F}_{q^m}$ is\[\sum_{\mathbf{y}\in\mathbb{F}_{q^m}^{km}} w_H(\mathbf{y})=km\cdot q^{m(km-1)}(q^m-1)=M\cdot km\cdot (1-q^{-m}).\]Hence,\begin{equation}
   \sum_{\mathbf{u}} w_{srk}(\mathbf{u})\le M\cdot km\cdot (1-q^{-m}).
   \label{eqn7}
   \end{equation}
   Substituting ~\eqref{eqn7} into ~\eqref{eqn6} and using $|\ker(f)|=q^{m^2(k-l)}$, we obtain
   \[S_1\ge(2t+1)\bigl(M-q^{m^2(k-l)}\bigr)-M\cdot km\cdot(1-q^{-m})+s_{srk}.\]
   Therefore,\begin{equation}
   \begin{aligned}
   \sum_{i,j}[\mathbf{D}]_{ij} = M \cdot S_1
   \ge M\Bigl((2t+1)\bigl(M-q^{(k-l)m^2}\bigr) \\
   - M\cdot km\cdot(1-q^{-m})+s_{\mathrm{srk}}\Bigr).
   \end{aligned}
   \label{eqn8}
   \end{equation}
   From ~\eqref{eqn5} and ~\eqref{eqn8},\begin{equation}
   \begin{aligned}
   \frac{rm}{q^m}\cdot M^2(q^m-1)\ge M\Bigl((2t+1)\bigl(M-q^{(k-l)m^2}\bigr) \\
   -M\cdot km\cdot(1-q^{-m})+s_{\mathrm{srk}}\Bigr). 
   \end{aligned}
   \label{eqn9}
   \end{equation}
   Dividing ~\eqref{eqn9} both sides by $M^2$, we obtain \[\frac{rm}{q^m}(q^m-1)\ge(2t+1)(1-q^{-l m^2})-km(1-q^{-m})+\frac{s_{\mathrm{srk}}}{M}\]
   Since $M=q^{km^2}$, solving for $r$ gives \begin{align*}r\ge\frac{q^m}{m(q^m-1)}\Bigl((2t+1)(1-q^{-l m^2}) - km(1-q^{-m}) \\+\frac{s_{\mathrm{srk}}}{q^{km^2}}\Bigr).
   \end{align*}
   \[r\ge\Bigl(\frac{q^m}{m(q^m-1)}\Bigr)(2t+1)(1-q^{-m^2l})-k \\+\frac{s_{\mathrm{srk}}}{m(1-q^{-m})q^{m^2k}}.
    \]
\end{IEEEproof}
\begin{rem}
\label{rem5}
    Let $f:\left(\mathbb{F}_q^{m\times m}\right)^k \to \left(\mathbb{F}_q^{m\times m}\right)^k$ be a bijective linear function. Then $\ker(f)=\{0\}$, so $s_{\mathrm{srk}}=0$. Moreover, $l=k$, and hence $q^{-l m^2}=q^{-k m^2}$.
    Substituting these into Theorem \ref{theorem4} yields
    \[r \geq \frac{q^{m}}{m(q^{m}-1)}(2t+1)\left(1-q^{-k m^2}\right)-k,\]
    which matches the Plotkin-like bound for ECCs for sum-rank metric given in Remark ~\ref{rem3} and also with \cite{b8}, i.e.,\[n:=k+r_{srk}^f(m,k,t) \geq \frac{q^{m}}{m(q^{m}-1)}(2t+1)\left(1-q^{-k m^2}\right).\]
\end{rem}
\begin{rem}
\label{rem6}
    For $m=1$, the sum‑rank metric coincides with the Hamming metric. The Plotkin bound for linear functions for sum-rank metric given in Theorem ~\ref{theorem4} reduces to Plotkin bound for linear functions for Hamming metric given in \cite{b12}.
\end{rem}
\begin{example}
\label{example4}
   Let $q=2$, $m=2$, $k=2$, and $\mathbf{u} \in \left(F_q^{m\times m}\right)^k$. Consider a linear function $f:(F_q^{m\times m})^k \to (F_q^{m\times m})^l$ defined as $f(\mathbf{u})=\sum_{i=1}^{k}u_i$, where $u_i \in F_q^{m\times m}$. Thus $l=1$, $\ker(f)=\{\mathbf{u}:f(\mathbf{u})=0\}$, and $|\ker(f)|=q^{m^2(k-l)}=16$. It gives $s_{\mathrm{srk}}=\sum_{\mathbf{u}\in\ker(f)} w_{\mathrm{srk}}(\mathbf{u})=42$. On mapping $\Phi: (F_2^{2\times2})^2 \to F_4^4$, it gives $\sum_\mathbf{u} w_{\mathrm{srk}}(\mathbf{u})=672 < \sum_{\mathbf{y}\in F_4^4} w_H(\mathbf{y})=768$. Substituting this into Theorem 4 with $t=2$ gives the lower bound $r\ge2$.
\end{example}
\begin{example}
\label{example5}
    Let $q=2$, $m=2$, $k=l=1$, and consider the linear function (Bijective) $f:F_q^{m\times m}\to F_q^{m\times m}$ defined as $f(\mathbf{X})=\mathbf{AX}$, where $\mathbf{A}=\begin{bmatrix}1 & 1\\0 & 1\end{bmatrix}$. Since $\mathbf{A}$ is full rank, $\ker(f)=\{0\}$ and hence $|\ker(f)|=1$. Therefore, $s_{\mathrm{srk}}=\sum_{\mathbf{u}\in\ker(f)} w_{\mathrm{srk}}(\mathbf{u})=0$. Further, $\sum_{\mathbf{X}\in F_2^{2\times2}}rk(\mathbf{X})=21$ and $\sum_{\mathbf{y}\in F_4^2} w_H(\mathbf{y})=24$. Substituting these into Theorem 4 with $t=1$ gives $r \ge 1$.
\end{example}

\section{Sum-rank Locally Binary Functions}
\label{Sum-rank locally binary Functions}
In this section, we study FCSRCs for sum-rank locally binary functions. We derive optimal redundancy and show how it can be achieved using a simple explicit code construction.
\begin{defn}
\label{function ball}
    The function ball of a function $f$ with radius $\rho$ around $\mathbf{u} \in (\mathbb{F}_q^{m \times m})^k$ is defined by
    \[B_{srk}^f(\mathbf{u}, \rho) = \{ f(\mathbf{u}') \; : \; \mathbf{u}' \in (\mathbb{F}_q^{m \times m})^k \ \wedge \ d_{srk}(\mathbf{u}, \mathbf{u}') \leq \rho \}.\]

\end{defn}
\begin{defn}
\label{loc bin}
 A function $f : (\mathbb{F}_q^{m \times m})^k  \to Im(f)$ is called a $\rho$-sum-rank locally binary function if, for all $\mathbf{u} \in (\mathbb{F}_q^{m \times m})^k$, we have $|B_{srk}^f(\mathbf{u}, \rho)| \leq 2.$
\end{defn}

Now, we present an explicit construction with which we prove the optimal redundancy of FCSRCs for $2t$-sum-rank locally binary function.

\begin{cons}
\label{Cons1}
    For $\mathbf{u} \in (\mathbb{F}_q^{m \times m})^k$, and for a $2t$-sum-rank locally binary function, define 
    \[w_{2t}(\mathbf{u}) = \begin{cases} \mathbf{I}_m, & \text{if } f(\mathbf{u}) = \max B_{srk}^f(\mathbf{u}, 2t), \\\mathbf{0}_m, & \text{otherwise,}\end{cases}\]
    Then $\mathrm{Enc(\mathbf{u})} = \bigl(\mathbf{u}, (w_{2t}(\mathbf{u}))^{\lceil\frac{2t}{m}\rceil} \bigr),$ 
where $(w_{2t}(\mathbf{u}))^{\lceil\frac{2t}{m}\rceil}$ denotes the $\lceil\frac{2t}{m}\rceil$-fold repetition of $w_2t(\mathbf{u}).$

\end{cons}
\begin{lem}
\label{lemma1}
For any $2t$-sum-rank locally binary function $f$, we have $r_{srk}^f(m,k,t) = \left\lceil \frac{2t}{m} \right\rceil.$

\end{lem}
\begin{IEEEproof}
    Let $\mathbf{u}, \mathbf{u}' \in (\mathbb{F}_q^{m \times m})^k$ such that $f(\mathbf{u}) \ne f(\mathbf{u}')$. If $d_{srk}(\mathbf{u},\mathbf{u}') \ge 2t$, then $d_{srk}(\mathrm{Enc(\mathbf{u}),Enc(\mathbf{u}')}) \geq 2t+1$. If $d_{srk}(\mathbf{u},\mathbf{u}') \leq 2t$, then from Construction~\ref{Cons1}, $w_{2t}(\mathbf{u}) \neq w_{2t}(\mathbf{u}'). \; \mathrm{Enc(\mathbf{u})} =(\mathbf{u},w_{2t}(\mathbf{u})), \mathrm{Enc(\mathbf{u}')}=(\mathbf{u}',w_{2t}(\mathbf{u}')).$
    $d_{srk}(\mathrm{Enc(\mathbf{u}), Enc(\mathbf{u}')}) = d_{srk}(\mathbf{u}, \mathbf{u}') + d_{srk}(p(\mathbf{u}), p(\mathbf{u}')).$
    $d_{srk}(\mathbf{u},\mathbf{u}') \geq 1,\; d_{srk}(w_{2t}(\mathbf{u}),w_{2t}(\mathbf{u}')) = m\lceil\frac{2t}{m}\rceil \geq m\frac{2t}{m} = 2t,\; d_{srk}(\mathrm{Enc(\mathbf{u}),Enc(\mathbf{u}')}) \geq 2t+1$.
\end{IEEEproof}
\begin{example}
\label{example6}
	Let $q = 2, k = 2, m = 2, t = 1, \mathbf{u},\mathbf{u}' \in ( \mathbb{F}_2^{2 \times 2})^2$ with 
	\[
	\mathbf{u} =
	\left[
	\begin{pmatrix}
		1 & 0\\
		0 & 0
	\end{pmatrix},
	\begin{pmatrix}
		1 & 0\\
		0 & 1
	\end{pmatrix}
	\right],
	\quad
	\mathbf{u}' =
	\left[
	\begin{pmatrix}
		1 & 0\\
		0 & 0
	\end{pmatrix},
	\begin{pmatrix}
		1 & 1\\
		0 & 0
	\end{pmatrix}
	\right].
	\] We define the function 
	$f :  (\mathbb{F}_2^{m \times m})^k \to \mathbb{F}_2$
	by
	$
	f(\mathbf{u}) = \sum_{i=1}^{k} \det\!\left(u^{(i)}\right),
	$
	where the summation is over \( \mathbb{F}_2 \), \( f(\mathbf{u}) \in \{0,1\} \).
	Thus, $f(\mathbf{u}) = 1, f(\mathbf{u}') = 0, d_{srk}(\mathbf{u},\mathbf{u}') = 1$. From Construction~\ref{Cons1}, it follows that $w_{2t}(\mathbf{u}) = \mathbf{I}_2, w_{2t}(\mathbf{u}') = \mathbf{0}_2$.
	\[
	\mathrm{Enc}(\mathbf{u}) =
	\left[
	\begin{pmatrix}
		1 & 0\\
		0 & 0
	\end{pmatrix},
	\begin{pmatrix}
		1 & 0\\
		0 & 1
	\end{pmatrix},
	\begin{pmatrix}
		1 & 0\\
		0 & 1
	\end{pmatrix}
	\right]\]
	\[\mathrm{Enc}(\mathbf{u}') =
	\left[
	\begin{pmatrix}
		1 & 0\\
		0 & 0
	\end{pmatrix},
	\begin{pmatrix}
		1 & 1\\
		0 & 0
	\end{pmatrix},
	\begin{pmatrix}
		0 & 0\\
		0 & 0
	\end{pmatrix}
	\right].
	\]$d_{srk}(\mathrm{Enc(\mathbf{u}),Enc(\mathbf{u}')})=3$.
\end{example}
\section{Sum-rank weight Functions}
\label{Sum-rank weight Functions}
In this section, we study a  class of functions called sum-rank weight functions. we begin with the definition of such functions.
\begin{defn}[Sum-rank weight function]
\label{srk weight}
 A sum-rank weight function is defined as $f(\mathbf{u}) = w_{srk}(\mathbf{u})$ where $\mathbf{u} \in (\mathbb{F}_q^{m \times m})^k$ and $m, k \in \mathbb{N}$.
\end{defn}
The expressiveness of sum-rank weight function is given by $E = |Im(w_{srk})| = km+1$. We now provide a construction of FCSRCs for the sum-rank weight functions. The construction is based on the idea of assigning same parity vectors to those information vectors with same function value $f(\mathbf{u})$ and different parity vectors to those with different $f(\mathbf{u})$.
\begin{cons}
\label{Cons2}
    For $\mathbf{u} \in (\mathbb{F}_q^{m \times m})^k$, $t,m,k \in \mathbb{N}$, we define $\mathrm{Enc}(\mathbf{u}) = (\mathbf{u}, \mathbf{p}_{w_{srk}(\mathbf{u})+1})$ where $p_i$'s are defined as follows:
    Let $\mathcal{P} = \{\mathbf{p}_1,\mathbf{p}_2,\mathbf{p}_3,...,\mathbf{p}_{2t+1}\}$ be an srk code with minimum distance $d_{min}(\mathcal{P})\;=\;\min_{\substack{\mathbf{p}_i, \mathbf{p}_j \in \mathcal{P}}}d_{srk}(\mathbf{p}_i,\mathbf{p}_j) = 2t\; \text{for all}\; i, j \leq 2t+1$ and $i \neq j.$ Set $\mathbf{p}_i = \mathbf{p}_{i\;smod\;(2t+1)}$ for $i \geq 2t+2$ where $a\;smod\;b = ((a - 1)\;mod\;b) + 1$.\\ From Definition ~\ref{MRD}, we define a constant maximum rank distance code $\mathcal{C}$ with cardinality $q^m$, where $q$ is a prime power and minimum distance is $m$. We define the construction of parity code for $q,m,t$ such that $q^m \ge 2t$ as follows:\\
    The sum-rank metric code $\mathcal{P}$ has codewords $\mathbf{p}_i$ with matrices from $\mathcal{C}$ rank metric code. If $ q^m \ge 2t+1,\mathbf{p}_i =\{R_i\}^{\left\lceil \frac{2t}{m} \right\rceil}, \; for\;i=1,2,..,2t+1$. If $q^m = 2t,$
    \[\mathbf{p}_i =\begin{cases} \{R_i\}^{\left\lceil \frac{2t}{m} \right\rceil}, & for\;i=1,2,..,2t,\\ \mathbf{p}_{i-2t}+\mathbf{m}, &for\;i=2t+1,..,4t.\end{cases}\]
\end{cons}
\noindent where $\mathbf{p}_i = \{R_i\}^{\left\lceil \frac{2t}{m} \right\rceil}$ which denotes the $\left\lceil \frac{2t}{m} \right\rceil$-fold repetition of $R_i \in \mathcal{C}$. Here $\mathbf{m} = (0_m,0_m,..,E_m) \in (\mathbb{F}_q^{m\times m})^r$ with $E_m \in \mathbb{F}_q^{m \times m}$ as a rank-1 matrix such that $d_{srk}(\mathbf{p}_{2t},\mathbf{p}_{2t+1}) = 2t.$
We show that this encoding function gives an FCSRC for sum-rank weight functions using the following lemma.

\begin{lem}
\label{lemma2}
Construction ~\ref{Cons2} gives an FCSRC for the sum-rank weight function $f(\mathbf{u}) = w_{srk}(\mathbf{u})$ that can correct $t$ errors. The optimal redundancy of the code is 
$\left\lceil \frac{2t}{m} \right\rceil$, for $k \in \mathbb{N}$ and for all $m,t,q$ which satisfy $q^m\ge2t$.
\end{lem}
\begin{IEEEproof}
    Let $\mathbf{u}_i,\mathbf{u}_j \in (\mathbb{F}_q^{m \times m})^k$ with $f(\mathbf{u}_i) \neq f(\mathbf{u}_j)$. If $|w_{srk}(\mathbf{u}_i) - w_{srk}(\mathbf{u}_j )| \geq 2t + 1$, then $d_{srk}(\mathrm{Enc}(\mathbf{u}_i), \mathrm{Enc}(\mathbf{u}_j )) \geq d_{srk}(\mathbf{u}_i, \mathbf{u}_j )\geq |w_{srk}(\mathbf{u}_i) - w_{srk}(\mathbf{u}_j )| \geq 2t + 1.$ If $|w_{srk}(\mathbf{u}_i) - w_{srk}(\mathbf{u}_j )| \leq 2t $, then $d_{srk}(\mathrm{Enc}(\mathbf{u}_i), \mathrm{Enc}(\mathbf{u}_j )) = d_{srk}(\mathbf{u}_i, \mathbf{u}_j ) + d_{srk}(\mathbf{p}_i,\mathbf{p}_j).$ Since $\mathbf{u}_i \neq \mathbf{u}_j$, we have $d_{srk}(\mathbf{u}_i,\mathbf{u}_j) \geq 1$. Further, by the construction of the code $P$, we have $d_{srk}(\mathbf{p}_i,\mathbf{p}_j) \geq 2t$. Hence, $d_{srk}(\mathrm{Enc}(\mathbf{u}_i), \mathrm{Enc}(\mathbf{u}_j))\geq 2t+1.$\\
    If $q^m \ge 2t+1$, We define $\mathbf{p}_i = \{R_i\}^{\left\lceil \frac{2t}{m} \right\rceil}$ which denotes the $\left\lceil \frac{2t}{m} \right\rceil$-fold repetition of $R_i \in \mathcal{C}$,  for $i=1,2,..,2t+1$. Since any pair of matrices in $\mathcal{C}$ has rank distance $m$, the sum-rank distance between any two parity vectors $\mathbf{p}_i$ and $\mathbf{p}_j$ is \[ d_{srk}(\mathbf{p}_i,\,\mathbf{p}_j) \ge m \left\lceil \frac{2t}{m} \right\rceil \ge 2t.\]. Thus, we can construct $2t+1$ parity vectors with minimum sum-rank distance $2t$ using $r = \left\lceil \frac{2t}{m} \right\rceil$ blocks. For $q^m = 2t$ case, we cannot construct $2t+1$ distinct parity vectors using only distinct elements from code $\mathcal{C}$. From Construction \ref{Cons2}, consider the parity vectors $\mathbf{p}_i$ and $\mathbf{p}_j$ with $i \ne j$.
    \begin{itemize}
        \item Case 1:\\
        1a) If $i,j \leq 2t$, then $\mathbf{p}_i$ and $\mathbf{p}_j$ are $\left\lceil \frac{2t}{m} \right\rceil$-fold repetitions of distinct codewords from $\mathcal{C}$. Since $\mathcal{C}$ has minimum rank distance $m$, we have\[d_{srk}(\mathbf{p}_i,\mathbf{p}_j) \ge m \cdot \left\lceil\frac{2t}{m} \right\rceil \ge 2t.\]
        1b) If $i,j > 2t$, then $\mathbf{p}_i = \mathbf{p}_{i-2t} + \mathbf{m}$ and $\mathbf{p}_j = \mathbf{p}_{j-2t} + \mathbf{m}$. Since adding the same $\mathbf{m}$ to both vectors does not change their distance, we have\[d_{srk}(\mathbf{p}_i,\mathbf{p}_j) = d_{srk}(\mathbf{p}_{i-2t}, \mathbf{p}_{j-2t})\ge m \cdot \left\lceil \frac{2t}{m} \right\rceil \ge 2t.\]
        \item Case 2: $i \leq 2t, j > 2t$\\
        2a) If $j = i+2t$, then $\mathbf{p}_j = \mathbf{p}_i + \mathbf{m} = \mathbf{p}_{j-2t} + \mathbf{m}$, $d_{srk}(\mathbf{p}_i,\mathbf{p}_j) = d_{srk}(\mathbf{m}) = 1$. Now consider $|f(\mathbf{u}_i) - f(\mathbf{u}_j)| = |i - j|$. Then, $|f(\mathbf{u}_i) - f(\mathbf{u}_j)| = 2t$. We have $d_{srk}(\mathbf{u}_i, \mathbf{u}_j) \ge |f(\mathbf{u}_i) - f(\mathbf{u}_j)| = 2t.$ Thus $d_{srk}(\mathrm{Enc}(\mathbf{u}_i), \mathrm{Enc}(\mathbf{u}_j))= d_{srk}(\mathbf{u}_i, \mathbf{u}_j) + d_{srk}(\mathbf{p}_i, \mathbf{p}_j) \ge 2t + 1.$\\
        2b) If $j \neq i + 2t$, then we have $\mathbf{p}_i = (A_1, A_2, \ldots, A_r)$ where $A_k = R_i \in \mathcal{C}$, and $\mathbf{p}_{j-2t} = (B_1, B_2, \ldots, B_r)$ where $B_k = R_{j-2t} \in \mathcal{C}$. Then $\mathbf{p}_j = \mathbf{p}_{j-2t} + \mathbf{m} = (B_1 + 0_m, B_2 + 0_m, \ldots, B_r + E_m)$. First $r - 1$ blocks are matrices that differ by rank $m$ as $d_{rk}(R_i,R_{j-2t})=m$ and the last block has rank distance $d_{rk}(A_r,B_r+E_m)\geq m-1$. Thus, sum-rank distance between parity vectors is $d_{srk}(\mathbf{p}_i,\mathbf{p}_j) \geq m(\lceil\frac{2t}{m}\rceil-1)+m - 1 \geq 2t-1$. Now two cases arise.
        \begin{enumerate}
        \item Casee $|i - j| = 1:$ Since $i \leq 2t$ and $j > 2t$, the only possibility is $i = 2t$ and $j = 2t+1$. By construction ~\ref{Cons2}, we have ensured $d_{srk}(\mathbf{p}_{2t}, \mathbf{p}_{2t+1}) = 2t$. As $|f(\mathbf{u}_i) - f(\mathbf{u}_j)| = 1$, we have $d_{srk}(\mathbf{u}_i, \mathbf{u}_j) \geq 1$. Thus, $d_{srk}(\mathrm{Enc}(\mathbf{u}_i), \mathrm{Enc}(\mathbf{u}_j))= d_{srk}(\mathbf{u}_i, \mathbf{u}_j) + d_{srk}(\mathbf{p}_i, \mathbf{p}_j) \geq 2t+1$.
        \item Case $|i - j| \geq 2:$ Then $d_{srk}(\mathbf{u}_i, \mathbf{u}_j) \geq |i - j| \geq 2$. $d_{srk}(\mathrm{Enc}(\mathbf{u}_i), \mathrm{Enc}(\mathbf{u}_j))= d_{srk}(\mathbf{u}_i, \mathbf{u}_j) + d_{srk}(\mathbf{p}_i, \mathbf{p}_j) \geq 2+2t-1=2t+1.$
        \end{enumerate}
    \end{itemize}
    From Corollary ~\ref{cor1}, for any function $f$ with $|\mathrm{Im}(f)| \ge 2$, the redundancy of an FCSRC satisfies \[ r^{f}_{srk}(m,k,t) \ge N_{srk}(2,2t)= \left\lceil \frac{2t}{m} \right\rceil .\]
    For both the cases, Construction ~\ref{Cons2} achieves optimal redundancy $r = \lceil\frac{2t}{m}\rceil$.
\end{IEEEproof}
The following two examples are given one  for the $q^m \geq 2t+1$ and the other for $q^m = 2t.$ 
\begin{example}
\label{example7}
Let $q=2, m=2, t =1,$ and  $\mathbf{u} \in (\mathbb{F}_2^{2 \times 2})^2$. Here $q^m = 4 \geq 2t+1 =3$. From Construction ~\ref{Cons2}, it requires $2t+1$ parity vectors with $d_{min}=2t$. They are $\mathbf{p}_1,\mathbf{p}_2,\mathbf{p}_3$ with $d_{srk}(\mathbf{p}_i,\mathbf{p}_j) \geq 2$ for $1 \leq i,j \leq 3$. The length of these parity vectors will be $r =\left\lceil \frac{2t}{m} \right\rceil = 1$.We consider a constant maximum rank metric code from Definition ~\ref{MRD} as \[
	\mathcal{C} =
	\left\{
	\begin{pmatrix}
		0 & 0 \\
		0 & 0
	\end{pmatrix},
	\begin{pmatrix}
		1 & 0 \\
		0 & 1
	\end{pmatrix},
	\begin{pmatrix}
		1 & 1 \\
		1 & 0
	\end{pmatrix},
	\begin{pmatrix}
		0 & 1 \\
		1 & 1
	\end{pmatrix}
	\right\}.
	\]. Define $\mathbf{p}_1,\mathbf{p}_2,\mathbf{p}_3$ from codewords of $\mathcal{C}$ with $r=1$ which gives pairwise sum-rank distance = 2 between them as required. Set $\mathbf{p}_i = \mathbf{p}_{i \;smod\; 3}$ for $i \ge 4$. This approach assigns these parity vectors to message vectors with different function values as $\mathrm{Enc}(\mathbf{u}) = (\mathbf{u}, \mathbf{p}_{w_{srk}(\mathbf{u})+1})$ irrespective of the value of $k$.
\end{example}
\begin{example}
\label{example8}
    Let $q=2, m=2,t =2,$ and $ \mathbf{u} \in (\mathbb{F}_2^{2 \times 2})^2$. This gives $q^m = 4 =2t$. In this case we consider $\mathbf{p}_1,\mathbf{p}_2,\mathbf{p}_3,\mathbf{p}_{4}$ as vectors with 
    $2$-fold repetition of $R_i$ where $R_i \in \mathcal{C}$ defined in example~\ref{example7}.
    $\mathbf{p}_1 =
	\left[
	\begin{pmatrix}
		0 & 0\\
		0 & 0
	\end{pmatrix},
	\begin{pmatrix}
		0 & 0\\
		0 & 0
	\end{pmatrix}
	\right],
	\;
	\mathbf{p}_2 =
	\left[
	\begin{pmatrix}
		1 & 0\\
		0 & 1
	\end{pmatrix},
	\begin{pmatrix}
		1 & 0\\
		0 & 1
	\end{pmatrix}
	\right],\; \mathbf{p}_3 =
	\left[
	\begin{pmatrix}
		1 & 1\\
		1 & 0
	\end{pmatrix},
	\begin{pmatrix}
		1 & 1\\
		1 & 0
	\end{pmatrix}
	\right],\; \mathbf{p}_4 =
	\left[
	\begin{pmatrix}
		0 & 1\\
		1 & 1
	\end{pmatrix},
	\begin{pmatrix}
		0 & 1\\
		1 & 1
	\end{pmatrix}
	\right]$. 
    Set $\mathbf{p}_i = \mathbf{p}_{i-4} + \mathbf{m}$ where $
	\mathbf{m} =
	\left[
	\begin{pmatrix}
		0 & 0\\
		0 & 0
	\end{pmatrix},
	\begin{pmatrix}
		0 & 0\\
		0 & 1
	\end{pmatrix}
	\right] $ for $i=5,6,7,8$. It satisies that $d_{srk}(\mathbf{p}_4,\mathbf{p}_5) = 4$. We get $
	\mathbf{p}_5 =
	\left[
	\begin{pmatrix}
		0 & 0\\
		0 & 0
	\end{pmatrix},
	\begin{pmatrix}
		0 & 0\\
		0 & 1
	\end{pmatrix}
	\right],
	\;
	\mathbf{p}_6=
	\left[
	\begin{pmatrix}
		1 & 0\\
		0 & 1
	\end{pmatrix},
	\begin{pmatrix}
		1 & 0\\
		0 & 0
	\end{pmatrix}
	\right],\;\mathbf{p}_7 =
	\left[
	\begin{pmatrix}
		1 & 1\\
		1 & 0
	\end{pmatrix},
	\begin{pmatrix}
		1 & 1\\
		1 & 1
	\end{pmatrix}
	\right],
	\;
	\mathbf{p}_8=
	\left[
	\begin{pmatrix}
		0 & 1\\
		  1 & 1
	\end{pmatrix},
	\begin{pmatrix}
		0 & 1\\
		1 & 0
	\end{pmatrix}
	\right].
	$ Then we define $\mathbf{p}_i=\mathbf{p}_{i\;smod \;8}$ for $i \geq 9$. Assigning these parity vectors to message vectors ensures $d_{srk}(\mathrm{Enc}(\mathbf{u}), \mathrm{Enc}(\mathbf{u}')) \geq 2t+1=5$ for $f(\mathbf{u}) \neq f(\mathbf{u}')$.
\end{example}

\section{Conclusion}
\label{Conclusion}
We derived a Plotkin-like bound for irregular sum-rank distance codes and then presented explicit constructions of FCSRCs for sum-rank locally binary functions and sum-rank weight functions. We showed that the proposed construction achieve optimal redundancy for certain parameters. Further research directions include the study of FCSRCs for other classes of functions in the sum-rank metric as well as the development of tighter lower and upper bounds on the optimal redundancy.



\begin{thebibliography}{00}
	
	\bibitem{b1}
	A. Lenz, R. Bitar, A. Wachter-Zeh, and E. Yaakobi, ``Function-correcting codes,'' IEEE Transactions on Information Theory, vol. 69, no. 9, pp. 5604–5618, 2023.	
	
	\bibitem{b2}
    R. K\"{o}tter, R. Kschischang, ``Coding for Errors and Erasures in
    Random Network Coding,'' IEEE Transactions on Information Theory, vol. 54, no. 8, August 2008. 
	
	\bibitem{b3} 
	D. Silva, F. R. Kschischang, and R. Kötter, “A rank-metric approach to error control in random network coding”. IEEE Transactions on Information Theory, vol. 54, no. 9, September 2008.
	
	\bibitem{b4}
	E. M. Gabidulin, ``Theory of Codes with Maximal Rank Distance,'' Probl. Peredachi Inf., 1985, vol. 21, no. 1, pp. 3–16 [Probl. Inf. Trans. (Engl. Transl.), 1985, vol. 21, no. 1, pp. 1–12]
	
	\bibitem{b5}
	E. M. Gabidulin and N. I. Pilipchuk, ``Symmetric Rank Codes,'' Problems of Information Transmission, vol. 40, No. 2, pp. 103–117, 2004.

    \bibitem{b6}
    U. Mart\'{\i}nez-Pe\~{n}as, M. Shehadeh, and F. R. Kschischang. ``Codes in the Sum-Rank Metric, Fundamentals and Applications,'' Foundations and Trends in Communications and Information Theory, vol. 19, no. 5, pp. 814-1031, 2022.

    \bibitem{b7}
    E. Gorla, U. Mart\'{\i}nez-Pe\~{n}as, and Flavio Salizzoni. ``Sum-rank metric codes.'' Available on arXiv:2304.12095 [cs.IT], April 2023.
    
	\bibitem{b8}
	E. Byrne, H. Gluesing-Luerssen and A. Ravagnani, ``Fundamental properties of sum-rank-metric codes, IEEE Transactions on Information Theory, vol. 67, no. 10, pp. 6456-6475, 2021.
	
	\bibitem{b9}
	Q. Xia, H. Liu, and B. Chen, ``Function-correcting codes for symbol-pair read channels,'' IEEE Transactions on Information Theory, vol. 70, no. 11, pp. 7807-7819, 2023.

    \bibitem{b10}
    A. Singh, A. Kumar Singh, and E. Yaakobi, ``Function-correcting codes for b-symbol read channels,'' Available on arXiv:2503.12894 [cs.IT], March 2025.

    \bibitem{b11}
	S. Sampath, and B. S. Rajan, ``On Plotkin Bound for Function-Correcting Codes for b-Symbol Read Channels,'' 2025 IEEE Information Theory Workshop (ITW), Sydney, Australia, 2025, pp. 698-703.
    
	\bibitem{b12}
	R. Premlal and B. S. Rajan, ``On Function-Correcting Codes,'' in IEEE Transactions on Information Theory, vol. 71, no. 8, pp. 5884-5897, August 2025.
	
	\bibitem{b13}
	Y. Zhang, Z. Xu, X. Zhang and G. Ge, ``Optimal Redundancy of Function-Correcting Codes,'' in IEEE Transactions on Information Theory, vol. 71, no. 12, pp. 9458-9467, December 2025.
	
	\bibitem{b14}
	C. Rajput, B. S. Rajan, R. Freij-Hollanti and C. Hollanti, ``Function-Correcting Codes for Locally Bounded Functions,'' 2025 IEEE Information Theory Workshop (ITW), Sydney, Australia, 2025, pp. 851-856.
	
	\bibitem{b15}
	G. K. Verma, A. Singh and A. Kumar Singh, ``Function-Correcting b-Symbol Codes for Locally $(\lambda,\rho, b) $-Functions,'' in IEEE Transactions on Information Theory, vol. 72, no. 1, pp. 331-341, January 2026.
	
	\bibitem{b16}
	H. Ly and E. Soljanin, ``On the Redundancy of Function-Correcting Codes over Finite Fields,'' 2025 13th International Symposium on Topics in Coding (ISTC), Los Angeles, CA, USA, pp. 1-5, 2025.  
	
	\bibitem{b17}
	Huiying, and H. Liu, ``Function-Correcting Codes with Homogeneous Distance,'' Finite Fields and Their Applications 112 (2026) 102791.
	
	\bibitem{b18}
	G. K. Verma, and A. K. Singh, ``On Function-Correcting Codes in the Lee Metric,'' Available on arXiv:2507.17654 [cs.IT], July 2025. 

    \bibitem{b19}
    Hareesh K., Rashid Ummer N.T. and B. S. Rajan, ``Plotkin-like Bound and Explicit Function-Correcting Code Constructions for Lee Metric Channels,'' Available on arXiv:2508.01702v3 [cs.IT], October 2025 (Accepted for publication in IEEE Transactions on Information Theory.

    \bibitem{b20}
	A. Singh, and A. Kumar Singh, ``Function-Correcting Codes for Insertion-Deletion Channel,'' Available on arXiv:2512.07243  [cs.IT], December 2025.
    
	\bibitem{b21}
	C. Rajput, B. S. Rajan, R. Freij-Hollanti, and C. Hollanti, ``Function-Correcting Codes With Data Protection,'' IEEE Transactions on Information Theory, Vol.72, No.7, July 2026, pp. 4860-4880. 
	
	
	
\end{thebibliography}
\end{document}